\documentclass{article}
\usepackage[utf8]{inputenc}
\usepackage{graphicx}
\usepackage{xcolor}
\usepackage{float}
\usepackage{array}  
\usepackage{amsthm}
 \usepackage{amsfonts,amssymb}
\usepackage{amsmath}

\def\vt{\vartheta}
\def\vp{\varphi}

\def\a{\alpha}
\def\b{\beta}
\def\ep{\epsilon}

\def\g{\gamma}
\def\d{\delta}
\def\g{\gamma}

\theoremstyle{definition}
\newtheorem{definition}{Definition}

\newtheorem{prop}{Proposition}

\begin{document}

\title{Hilbert's energy-momentum tensor  extended}

\author{Yakov Itin\\
Mathematics Department, Jerusalem College of Technology; \\Jerusalem, Israel}
\maketitle
\begin{abstract}
A variational derivative of a Lagrangian with regard to the metric tensor is used in classical field models to define Hilbert's energy-momentum tensor for a matter field. In solid-state physics, constitutive relationships between fundamental field variables are a topic that is covered by a broad variety of models. In this context, a constitutive tensor of higher order replaces the of the second-order metric tensor.
For the classical field models of gravity and electrodynamics, a similar premetric description with a linear constitutive relation has recently presented.
In this paper, we analyze the extension of the Hilbert definition of the energy-momentum tensor to models with general linear constitutive law. Differential forms are required for the covariant treatment of integrals on a differential manifold. The Lagrangian, electromagnetic current, and energy-momentum current must all be represented as twisted 4-forms, 3-forms, and vector-valued 3-forms, respectively. For an arbitrary linear map on forms, we derive a commutative variation identity that allows direct variation procedures without having to deal with the individual components.
One can deal with Maxwell-type Lagrangians in any dimension by restricting the linear map to the generalized Hodge dual map (constitutive law).
The Hilbert energy-momentum current, which is defined as a variation derivative of the Lagrangian with regard to a coframe field, is derived in differential form. It is demonstrated that the commutative variation identity is closely connected to the explicit form of the energy-momentum current. This construction is applied to a number of field models having a general linear constitutive law.
\end{abstract}

\section{Introduction}
Variation of a Lagrangian is a central procedure in both classical and quantum field theory. In addition to the field equations, it provides various  conservation quantities, in particular, the energy-momentum tensor.  When the action is expressed as an integral over a manifold $M$
\begin{equation}
    S=\int_M L\sqrt{-g}d^4x.
\end{equation}
Hilbert's energy-momentum tensor for a matter field is defined by a variation derivative expression 
\begin{equation}\label{Hilbert}
    T^{ij}=\frac 2{\sqrt{-g}}\frac{\d L\sqrt{-g}}{\d g_{ij}}.
\end{equation}
Here $L$ denotes a Lagrangian scalar,  $g_{ij}$ is a metric tensor with the determinant $g$.  
For non-gravity models that do not involve metric tensor derivatives, the variation derivative becomes the ordinary partial derivative.  Although its geometrical and physical interpretations, as well as the accompanying conservation law, are obtained only from Noether's technique, the formula (\ref{Hilbert}) is very useful for explicit computation.
Already the leading coefficient $(2/{\sqrt{-g}})$ implies that a more fundamental expression can be hidden in  (\ref{Hilbert}). 

For an integral defined on a differential manifold  \cite{deRham}, \cite{Cartan}, \cite{Spivak}, \cite{Bott}, an appropriate expression for the action functional is expressed by an integral
\begin{equation}\label{lagr}
    S=\int_M{\cal L},
\end{equation}
where the integrand ${\cal L}$ is assumed to be a twisted 4-form---{\it Lagrange density}. When the manifold is considered to be endowed with a coframe field $\vt^\a$ (a set of four linearly independent 1-forms), the metric can be expressed as a tensor product
\begin{equation}\label{metric}
    g=\eta_{\a\b}\vt^\a\otimes\vt^\b,
\end{equation}
where $\eta_{\a\b}={\rm diag}(1,-1,-1,-1)$ is the Minkowsky metric defined on a tangent space. 
Then, instead of considering a variation derivative of the scalar Lagrangian with regard to the metric tensor as in (\ref{Hilbert}), examine a variation of the Lagrangian density (\ref{lagr}) with respect to the coframe field 
\begin{equation}\label{Hilbert-Sigma}
    \Sigma_\a=\frac{\d{\cal L}}{\d\vt^\a}.
\end{equation}
The twisted vector-valued 3-form $\Sigma_\a$ is better suited for space-time integration over hypersurfaces than the tensor $T^{ij}$.  We will refer to $\Sigma_\a$ as {\it Hilbert's energy-momentum form}. This expression was widely used in mathematical physics; see, for example,  \cite{Truesdell}, \cite{Gotay}, \cite{Thirring}, \cite{Frankel}, and the references cited therein. 
The differential-form approach has been demonstrated to be useful in GR and its numerous modifications \cite{Hehl-RMP},    \cite{Mielke}, \cite{MAG}, \cite{Blagojevic-Hehl}. 

In this study, we seek a generalization of Hilbert's energy-momentum tensor for models in which the metric tensor is replaced by a general constitutive tensor. These premetric models are useful for a covariant description of electromagnetism in solid-state physics \cite{Post}, \cite{Kovetz}, \cite{Lindell}. Another relevant subject is classical electrodynamics and gravity with a general linear constitutive law, see \cite{Hehl-Obukhov},   \cite{Itin:2001bp},\cite{Itin:2001xz}, \cite{Itin:2016nxk},
\cite{Itin:2018dru}.

To extend Hilbert's notion to the premetric case, we consider a general constitutive law involving two fields  ${\cal H}$ and ${\cal F}$,  
\begin{equation}
    {\cal H}=\kappa{\cal F}.
\end{equation}
Here $\kappa$ is a pseudo-tensor that can be expressed in term of the basis $\vt^\a$. Then the definition (\ref{Hilbert-Sigma}) is applicable. 
To deal with the variation of the constitutive tensor, we extend the master formula \cite{Muench:1998ay} for a general constitutive map. 
Then we show that Hilbert's energy-momentum current is closely related to this commutative variation formula.

The paper is organized as follows: In Section 2, we present background information on the differential form formalism and variation operator. The general constitutive map is then defined, and its relationship to the variation technique is discussed. We derive the commutative relation for constitutive map variation. 
In Section 3, we apply the commutative relation for a  Maxwell-type model with a general constitutive law. The expression of the energy-momentum current 3-form and its relationship to the standard energy-momentum tensor is derived. The formulae for the trace of the energy-momentum current and its symmetric part are discussed.
In Section 4, we apply the formalism to  viable physics models, such as scalar fields, vacuum electrodynamics, premetric electromagnetism and premetric gravity. In Consequence section, we outline our results and explore some potential extensions. 

Notations: For representation of differential forms, we use a coframe field $\vt^\a$ and its dual frame field $e_\a$ with Greek indices varying in the range $\a,\b,\cdots=0,\cdots,(n-1)$. The coordinate frame is endowed with Roman indices $i,j,\cdots$ from the same range. Everywhere, Einstein's summation rule is assumed.

\section{Variations of differential forms and constitutive map}

In order to deal with the variation of differential forms, we need to make some assumptions about the nature of physical measuring quantities. Take the vector field $\cal A$, which we like to think of as a 1-form. In a holonomic (coordinate) basis $dxi$ and an anholonomic basis $varthetaa$, respectively, this form can be given in two different ways
\begin{equation}
    {\cal A}=A_idx^i=A_\a \vt^\a
\end{equation}
Which quantities—the components $A_i$, the components $A_\a$, or the differential form ${\cal A}$ itself—represent the proper physics field? Since every measurable quantity is actually an integral taken over some portion of the space and some interval of time, the differential form is the most appropriate quantity.
Every differentiable form contains information about both the geometry of the underline space and the physical field.
\subsection{Operations on differential forms} 
We start with a brief account of operators on differential forms. Precise definitions and proofs can be find eg in \cite{deRham}, \cite{Frankel},  \cite{Cartan},\cite{Spivak}, \cite{Bott}. For applications of differential-form formalism in physics, especially in electromagnetism and gravity field models, see \cite{Hehl-Obukhov} and \cite{MAG}.

We consider a differential manifold ${\cal M}$ of  dimension $n$ that is endowed with a coframe field (a set of $n$ smooth independent  1-forms)  $\vt^\a$.  Let $\Omega^p$ denotes the bundle of differential $p$-forms on ${\cal M}$.  The set $\vt^\a$ then serves as a basis for  $\Omega^1$, the set $\vt^\a\wedge\vt^\b$  serves as a basis for $\Omega^2$, and so on.

The {\it exterior product} of differential forms ${\cal A}]\in \Omega^k$ and ${\cal B}\in \Omega^\ell$  is a linear map 
\begin{equation}\label{2.1--1}
    \wedge:\Omega^k \times \Omega^\ell\to\Omega^{k+\ell}, 
\end{equation}
that satisfies the graded commutative law: 
\begin{equation}\label{2.1--2}
    {\cal A}\wedge{\cal B}=(-1)^{k+\ell}{\cal B}\wedge{\cal A}. 
\end{equation}
The {\it interior product} of a vector field $X$ and a differential form ${\cal A}$ is a linear map \footnote{Recall an alternative widely-used notation for the inner product,  $i_X\omega $. To our opinion, the notation  $X\rfloor\omega$ is  more convenient for lengthy expressions.   }  
\begin{equation}\label{2.1--3}
    X\rfloor:\Omega^k \to\Omega^{k-1}
\end{equation}
such that 
\begin{equation}\label{2.1--4}
    X\rfloor(Y\rfloor{\cal A})=-Y\rfloor(X\rfloor{\cal A}),
\end{equation}
and 
\begin{equation}\label{2.1--5}
    X\rfloor({\cal A}\wedge{\cal B})=(X\rfloor{\cal A})\wedge{\cal B}+(-1)^k{\cal A}\wedge(X\rfloor{\cal B}).
\end{equation}
It is convenient to use the dual {\it frame field } $e_\a$ along with a coframe field $\vt^\a$.  In terms of internal product, the duality equation is stated as 
\begin{equation}\label{2.1--6}
    e_\a\rfloor\vt^\b=\d_\a^\b.
\end{equation} 

We distinguish between the ordinary {\it untwisted forms}, which do not change when arbitrary basis transformations are used, and the {\it twisted forms}, which do change under  improper transformations (those with a negative determinant).  For detailed  description of these notions, see  \cite{deRham}.  Then in the standard basis an expression of a $p$-form ${\cal A}$ is presented as, 
\begin{equation}\label{base-express}
    {\cal A}=\frac 1{k!}A_{\a_1\cdots\a_k}\vt^{\a_1}\wedge \cdots\wedge \vt^{\a_k}\,.
\end{equation}
The coefficients $A_{\a_1\cdots\a_p}$ represent here an ordinary tensor for an untwisted form ${\cal A}$ and a pseudo-tensor for a twisted form ${\cal A}$. 
For a covariant approach for integration on a manifold, the use of twisted differential form as an appropriate integrand in physical meaningful integral expressions is essential.
In particular, a volume element (density) on $cal M$ is  defined as an arbitrarily smooth twisted $n$-form. We will  refer to a volume element as ${\rm ``vol"}$. In terms of  a prescribed volume element, an arbitrary twisted $k$-form can be expressed as 
\begin{equation}\label{base-express1}
    {\cal A}=\frac 1{k!}A^{\a_1\cdots\a_k}e_{\a_1}\rfloor\cdots e_{\a_k}\rfloor {\rm vol}.
\end{equation}
The collection of coefficients $A^{\a_1\cdots\a_k}$ in this formula, in contrast to (\ref{base-express}), represents an ordinary tensor.
Observe  useful interior form relations: 
 \begin{prop}
 For an arbitrary twisted (untwisted) $k$-form ${\cal A}$ on an $n$-dimensional manifold ${\cal M}$, 
\begin{equation}\label{2.1--7}
    e_\a\rfloor(\vt^\b\wedge {\cal A})={\cal A}\d_\a^\b-\vt^\b\wedge(e_\a\rfloor{\cal A})
\end{equation}
and
\begin{equation}\label{2.1--8}
    \vt^\a\wedge( e_\a\rfloor{\cal A})=k{\cal A},\qquad \qquad e_\a\rfloor(\vt^\a\wedge {\cal A})=(n-k){\cal A}\,.
\end{equation}
 \end{prop}
 \begin{proof}
 Eq.(\ref{2.1--7}) is a straightforward consequence of the rule (\ref{2.1--5}). 
 The first of Eqs.(\ref{2.1--8}) is derived by linearity:
 \begin{eqnarray}
\vt^\a\wedge( e_\a\rfloor{\cal A})&=&\frac 1{k!}A_{\a_1\cdots\a_p}\vt^\a\wedge\left(\d^{\a_1}_\a\vt^{\a_2}\wedge \cdots\wedge \vt^{\a_p} -\d^{\a_2}_\a\vt^{\a_1}\wedge \cdots\wedge \vt^{\a_p}+\cdots\right)      
 \nonumber\\
 &=&\frac k{k!}A_{\a_1\cdots\a_k}\vt^{\a_1}\wedge \cdots\wedge \vt^{\a_k}=k{\cal A}.
 \end{eqnarray}
 Then, the second of Eqs.(\ref{2.1--8}) follows from  Eq.(\ref{2.1--7}) by the contraction  in the indices $\a$ and $\b$.
 \end{proof}
Observe that for a top form ${\cal A}$ of the order $k=n$ Eq.(\ref{2.1--7}) yields
\begin{equation}\label{2.1--9}
    \vt^\b\wedge(e_\a\rfloor{\cal A})={\cal A}\d_\a^\b
\end{equation}
   \subsection{Constitutive map}
   In a wide range of physics models, especially in solid-state physics, the basis dynamical fields arice in pairs that are connected by a linear relation. 
The coefficients appearing in these {\it constitutive relation} express the fundamental media parameters.
 The constitutive relation construction has been successfully used in  classical electrodynamics \cite{Hehl-Obukhov} and even in gravity \cite{Itin:2018dru}.
   
   We start with a rather general definition of the constitutive law restricted only by linearity. We do not require, in particular, that the constitutive map be specified on differential forms of any arbitrary order. This is consistent with the concept of solid-state physics, where, for example, the dielectric tensor only correlates the vectors of the electric induction and the electric field and is not applicable to tensors of higher order.
\begin{definition}
{\it Constitutive map} $\diamond$ is defined as 
a general linear map on a $p$-form , 
\begin{equation}\label{2.3--1}
    \diamond: \Omega^p\to \Omega^q.
\end{equation}
 In particular, for two  arbitrary $p$-forms ${\cal A},{\cal B}$, and an arbitrary real  function (0-form) $f$ the following rules hold 
\begin{equation}\label{2.3--2}
    \diamond\,({\cal A}+{\cal B})=\diamond\,{\cal A}+\diamond\,{\cal B}\qquad {\rm and} \qquad \diamond(f{\cal A})=f\diamond{\cal A}\,.
\end{equation}
\end{definition}
The well-known example of a constitutive map is the Hodge  dual map $*$ that is defined on  differential forms (twisted and untwisted)  of an arbitrary order. 
On an $n$-dimensional manifold ${\cal M}$ endowed with a volume element ``${\rm vol}$'', the Hodge dual of an orthonormal coframe $\vt^\a$ satisfies the relations, see \cite{Hehl-Obukhov}, 
\begin{eqnarray}
*1&=&{\rm vol},\\
    *\vt^\a&=&\eta^{\a\mu}e_\mu\rfloor{\rm vol},\\
    *(\vt^\a\wedge\vt^\b)&=&\eta^{\a\mu}\eta^{\b\nu}e_\nu\rfloor e_\mu\rfloor{\rm vol},\qquad {\rm etc.} 
\end{eqnarray}
For an arbitrary $p$-form, Hodge map is defined by linearity,
\begin{equation}
    *{\cal A}=\frac 1{p!}A_{\a_1\cdots\a_p}*\left(\vt^{\a_1}\wedge \cdots\wedge \vt^{\a_p}\right).
\end{equation}
Observe  useful identities 
\begin{equation}
  \vt^\a\wedge *\vt^\b=\eta^{\a\b}\, {\rm vol}, \quad {\rm and}\quad  {\rm vol}=\frac 1n \eta_{\a\b}\vt^\a\wedge *\vt^\b
\end{equation}
that follow straightforwardly from the definition.

In order to deal with constitutive relations appearing in solid-state physics and premetric models, we extend the Hodge-dual map. 
\begin{definition}
The  linear map $\diamond$ will be called a {\it Hodge-type map} if it satisfies the conditions:
\begin{itemize}
    \item [(i)] It is a {\it dual} operator---for a $p$-form ${\cal F}$,  the form $\diamond{\cal F}$ is of the order $(n-p)$;
    \item [(ii)] It is a {\it twisted} (pseudo-tensorial) operator---the form $\diamond{\cal F}$ is twisted for an untwisted form ${\cal F}$ and vice versa;
    \item [(iii)] It is a {\it self-adjoint}  operator, i.e., for two $p$-forms ${\cal F},{\cal E}$
    \begin{equation}\label{3.1--1}
        \int{\cal F}\wedge\diamond\,{\cal E}=\int{\cal E}\wedge\diamond\,{\cal F}
    \end{equation}
\end{itemize}
\end{definition}
As a result, for any two arbitrary  $p$-forms ${\cal F},{\cal E}$,  the product ${\cal F}\wedge\diamond\,{\cal E}$ is a twisted $n$-form. Especially, the twisted Lagrangian $n$-form can be constructed as ${\cal L}\sim {\cal F}\wedge\diamond\,{\cal F}$.
\subsection{Variation operator on differential forms}
The Lagrangian $n$-form ${\cal L}$, as well as its ingredients (lower order differentials forms), are functions of  physical fields. These fields typically play a variety of roles in the variation procedure. The collection of physical variables can be separated into three subgroups within a given physics model:
\begin{itemize}
    \item [(i)] {\it Dynamical variables:} These variables are assumed to have a non-zero variation.  Field equations for these variables characterizing the dynamics of the physical fields are produced using the variation technique.
     \item [(ii)] {\it Semi-dynamical variables:} Despite the absence of field equations for these variables, it is presumed that semi-dynamical variables have non-zero variation.
     \item [(iii)] {\it Non-dynamical variables:} These variables are not included in any way in the variation procedure.
\end{itemize}
Let us give an illustration: In classical vacuum electrodynamics, the electric and magnetic fields are dynamical variables, whereas the electric current is a non-dynamical variable. The  metric tensor presents in the Lagrangian  as a semi-dynamical variable without its own field equation. Hilbert's expression for the electromagnetic energy-momentum tensor results from the variation of the metric tensor.

In order to account for variations in the physics field and independent variations in the geometric structure, we consider a differential manifold ${\cal M}$ of a dimension $n$  endowed with the coframe field $\vt^\a$.
The coframe variable will be regarded as semi-dynamical in the models that do not account for the effects of gravity. In gravity models, such as  telleparalel gravity, it is viewed as a fully dynamical variable. The coframe must additionally be confined to orthonormal in the standard metric GR.

Let us assume that the variation operator on differential forms fulfills the following rules:
\begin{definition}
The variation $\d$ of a $p$-form is defined as a smooth operator $\d:\Omega^p\to \Omega^p$ that satisfies the conditions:
\begin{itemize}
    \item [(1)] Linearity: for  (semi)dynamical $p$-forms ${\cal A, B}$ and non-dynamical scalar fields $\a,\b$,
    \begin{equation}\label{2.2--1}
        \d(\a {\cal A}+\b {\cal B})=\a \d {\cal A}+\b\d {\cal B};
        \end{equation}
    \item [(2)] For the wedge product of two (semi)dynamical forms the standard Lebniz rule holds
    \begin{equation}\label{2.2--2}
    \d({\cal A}\wedge {\cal B}) =(\d{\cal A})\wedge{\cal B} +{\cal A}\wedge (\d{\cal B});
    \end{equation}
\item [(3)] The exterior derivative of a (semi)dynamical form ${\cal A}$ commutes with the variation operator 
\begin{equation}\label{2.2--3}
    \d(d{\cal A})=d(\d{\cal A})\,.
\end{equation}
\item [(4)] For a (semi)dynamical 1-form ${\cal A}=A_\a \vt^\a$ with a (semi)dynamical coframe field $\vt^\a$,
\begin{equation}\label{2.2--4}
    \d{\cal A}=(\d A_\a)\vt^\a +A_\a(\d\vt^\a). 
\end{equation}
\end{itemize}
\end{definition}
As a result, we have in (\ref{2.2--4}) three variation quantities $\d{\cal A}, \d A_\a,$  and $\d\vt^\a$ that are connected by a single linear condition. One may select any two of them on their own. It is convenient to select the pair of differential forms $\{\d{\cal A},\d\vt^\a\}$ as independent variations in order to deal with the exterior forms formalism.

By using an interior product map taken with regard to the frame field $e_\a$, the relation (\ref{2.2--4}) can be equivalently stated for  $A_\a=e_\a\rfloor {\cal A}$ as 
\begin{equation}\label{2.2--5}
    \d{\cal A}=(\d A_\a)\vt^\a+\d\vt^\a \wedge (e_a\rfloor {\cal A}).
\end{equation}
Let us derive a generalization of this formula for an arbitrary $p$-form. 
\begin{prop}
Let a $p$-form ${\cal A}$ and a coframe field $\vt^\a$ both be assumed  dynamical (semi-dynamical) variables. Then  the variation of ${\cal A}$ takes the form 
\begin{equation}\label{2.2--6}
    \d{\cal A}=\frac 1{p!}(\d A_{\a_1\cdots\a_p})\,\vt^{\a_1}\wedge\cdots\wedge\vt^{\a_p}+\d\vt^\a\wedge (e_\a\rfloor A).
\end{equation}
\end{prop}
\begin{proof}
For a 2-form ${\cal A} =(1/2)A_{\a\b}\vt^\a\wedge \vt^\b$, we use the wedge product rule (\ref{2.2--2}) to express the variation as 
\begin{equation}
 \d{\cal A}=\frac 12\Big[(\d A_{\a\b})\vt^\a\wedge \vt^\b +A_{\a\b}\d(\vt^\a\wedge \vt^\b)\Big].
\end{equation}
The first term is as in (\ref{2.2--6}), while the second one is rewritten as 
\begin{eqnarray}
    && \frac 12 A_{\a\b}\big(\d(\vt^\a)\wedge \vt^\b+\vt^\a\wedge \d(\vt^\b)\big)=
     A_{\a\b}(\d\vt^\a)\wedge \vt^\b=\d\vt^\a\wedge\big(A_{\a\b}\vt^\b\big).
     \end{eqnarray}
     But
     \begin{equation}
        (e_\a\rfloor A)= \frac 12 A_{\mu\nu}(e_\a\rfloor (\vt^\mu\wedge\vt^\nu)=A_{\a\b}\vt^\b\,.
     \end{equation}
Hence, the relation (\ref{2.2--6}) holds for $p=2$. For  $p>2$, it follows by induction. 
\end{proof}

\subsection{Non-dynamical linear map and variation identity}
In this section, we consider a variation operator applied at a differential form transformed by a linear map. We assume that the variation of the transformed form $\diamond\,{\cal A}$  only consists of the variation of the coefficients of ${\cal A}$ and the variation of the coframe $\vt^\a$. Consequently,  for
the coframe field
\begin{equation}\label{2.3--3}
    \d(\diamond\vt^{\a})=\d\vt^\a\wedge(e_\mu\rfloor\diamond\vt^{\mu}),
\end{equation}
and for basis $p$-forms,
\begin{equation}\label{2.3--4}
     \d(\diamond\left(\vt^{\a_1}\wedge\cdots\wedge\vt^{\a_p}\right))=\d\vt^\a\wedge  (e_\a\rfloor\diamond\left(\vt^{\a_1}\wedge\cdots\wedge\vt^{\a_p}\right).
\end{equation}

The following identity provides a commutative relation for variation of a mapped form $\diamond{\cal A}$. A special case of this relation for the Hodge-dual map was provided in \cite{Muench:1998ay}. 
\begin{prop}
Let a $p$-form ${\cal A}$ and a coframe field $\vt^\a$ be (semi)dynamical variables. Let the variation of $\diamond\,\d{\cal A}$ consist from the variations of  ${\cal A}$ and of $\vt^\a$ only. Then, the commutative relation 
\begin{equation}\label{2.3--5}
    \boxed{\diamond\,\d{\cal A}-\d\diamond{\cal A}=
    \diamond\,(\d\vt^\a\wedge (e_a\rfloor {\cal A}))-\d\vt^\a\wedge  (e_a\rfloor {\diamond\,{\cal A}})}
\end{equation}
holds.
\end{prop}
\begin{proof}
Since the operators $\d$ and $\diamond$ are both linear, we have 
\begin{eqnarray}
\diamond\,(\d{\cal A})&=&
\diamond\left(\frac1{p!}\d{\cal A}_{\a_1\cdots\a_p}\vt^{\a_1}\wedge\cdots\wedge\vt^{\a_p}+\d\vt^\a\wedge (e_a\rfloor {\cal A})\right)\nonumber\\
      &=&\frac 1{p!}\d{\cal A}_{\a_1\cdots\a_p}\diamond \left(\vt^{\a_1}\wedge\cdots\wedge\vt^{\a_p}\right)+\diamond\big(\d\vt^\a\wedge (e_\a\rfloor {\cal A})\big).
\end{eqnarray}
For variation of the transformed form $\diamond\Omega$, we use (\ref{2.2--6}). 
Consequently,
\begin{eqnarray}
     \d(\diamond{\cal A})&=&\d\left(\frac 1{p!}{\cal A}_{\a_1\cdots\a_p}\diamond(\vt^{\a_1}\wedge\cdots\wedge\vt^{\a_p})\right)\nonumber\\&=&\frac 1{p!}\d{\cal A}_{\a_1\cdots\a_p}\diamond(\vt^{\a_1}\wedge\cdots\wedge\vt^{\a_p})+\d\vt^\a\wedge  (e_\a\rfloor \diamond{\cal A})).
\end{eqnarray}
Subtracting these two equations, we derive the identity (\ref{2.3--5}).
\end{proof}
A useful form of  identity (\ref{2.3--5}) is given by 
\begin{equation}\label{2.3--5x}
\boxed{ \d\diamond{\cal A}=\diamond\,\d{\cal A}-
    \diamond\,(\d\vt^\a\wedge (e_a\rfloor {\cal A}))+\d\vt^\a\wedge  (e_a\rfloor {\diamond\,{\cal A}})}
\end{equation}

Observe some simple consequences of this proposition:
\begin{itemize}
    \item [(1)] If ${\cal A}$  is a scalar, then $e_a\rfloor {\cal A}=0$, so 
\begin{equation}\label{mast1}
   \d\diamond{\cal A}= \diamond\d{\cal A}+ \d\vt^\a\wedge  (e_a\rfloor {\diamond{\cal A}})
\end{equation}
\item [(2)] Similarly, if $\diamond{\cal A}$ is a scalar, we have $e_a\rfloor {\diamond{\cal A}}=0$. Thus
\begin{equation}\label{mast2}
    \d\diamond{\cal A}=\diamond\d{\cal A}-
    \diamond(\d\vt^\a\wedge (e_a\rfloor {\cal A}))
\end{equation}
\item [(3)] When the variation of the coframe is presented by a 0-form matrix $\varepsilon^\a{}_\b$ of scalar parameters 
\begin{equation}
    \d\vt^\a=\varepsilon^\a{}_\b\vt^\b,
\end{equation}
Eq.(\ref{2.3--5x}) reads 
\begin{equation}\label{2.3--5xx}
 \d\diamond{\cal A}=\diamond\,\d{\cal A}-
    \varepsilon^\a{}_\b \left(\diamond\,(\vt^\b\wedge (e_a\rfloor {\cal A}))-\vt^\b\wedge  (e_a\rfloor {\diamond\,{\cal A}})\right)
\end{equation}
\item [(4)] For scalar variation of the coframe, $\d\vt^\a=\varepsilon \vt^\a $, the variation and constitutive map commute \begin{equation}\label{2.3--5xy}
\diamond\d{\cal A}=\d\diamond{\cal A}
\end{equation}
\end{itemize}
The linear operator $\diamond$ acts solely on the forms with the same order $p$ in all terms of  (\ref{2.3--5}). It means that even if the definition of the $\diamond$-operator  only applies to specific orders of forms, the formula is still correct.  We will see in the next section that this fact is  useful for physics applications.

\section{Variation of a Maxwell-type Lagrangian}
In this section we study how the variation commutative formula (\ref{2.3--5}) can be applied to a general Lagrangian of a Maxwell (Yang–Mills)  type.
\subsection{Generalized Maxwell-type Lagrangian}
We specialize the linear operator $\diamond$ to satisfy additional requirements of Definition 3 in order to build a Lagrangian from a $p$-form defined on an $n$-dimensional manifold.

We consider an action functional with  the {\it Maxwell-type Lagrangian} $n$-form
\begin{equation}\label{3.1--2}
    {\cal S}=\int{\cal L}=-\frac 12\int{\cal F}\wedge{\cal H} =-\frac 12\int{\cal F}\wedge\diamond\,{\cal F} 
\end{equation}
In this case, ${\cal F}$ is assumed to be an untwisted $p$-form, whereas the mapped form ${\cal H}=\diamond\,{\cal F}$ is  twisted and of the order $(n-p)$.

Variation of this  Lagrangian reads
\begin{equation}\label{3.1--3}
    \d{\cal S}=-\frac 12\int\d{\cal F}\wedge\diamond{\cal F}+{\cal F}\wedge\d\diamond{\cal F}.
\end{equation}
Using the commutative identity (\ref{2.3--5}), we rewrite the second term as
\begin{equation}\label{3.1--4}
    \d{\cal S}=-\frac 12\int\d{\cal F}\wedge\diamond{\cal F}+{\cal F}\wedge\Big(\diamond\d{\cal F}-
    \diamond(\d\vt^\a\wedge (e_a\rfloor {\cal F}))+\d\vt^\a\wedge  (e_a\rfloor \diamond {\cal F})\Big)
\end{equation}
Due to the self-adjoint property, 
\begin{equation}
    {\cal F}\wedge\diamond\d{\cal F}=\d{\cal F}\wedge\diamond{\cal F}.
\end{equation}
Consequently, 
\begin{equation}\label{3.1--5}
    \d{\cal S}=-\int\d{\cal F}\wedge\diamond{\cal F}+\d\vt^\a\wedge\frac 12\Big((e_a\rfloor {\cal F})\wedge\diamond{\cal F} -
    (-1)^p{\cal F}\wedge (e_a\rfloor \diamond{\cal F})\Big)
\end{equation}
We now identify the fundamental characteristic of the Maxwell (Yang-Mils) Lagrangian, namely the assumption that the field strength ${\cal F}$  is a closed form.
\begin{equation}\label{3.1--6}
    d{\cal F}=0.
\end{equation}
The form ${\cal F}$ is hence exact on a contractible domain, i.e., there is a  potential $(p-1)$-form ${\cal A}$ such that
\begin{equation}\label{3.1--7}
    {\cal F}=d{\cal A}.
\end{equation}
Another fundamental  assumption of Maxwell theory---the independent variations must be considered with respect to the potential ${\cal A}$. Then, the first term in  (\ref{3.1--5}) takes the form 
 \begin{equation}
     \d{\cal F}\wedge\diamond{\cal F}=
     \d (d{\cal A})\wedge\diamond{\cal F}=
     d (\d{\cal A})\wedge\diamond{\cal F}
 \end{equation}
Extracting the total derivative we are left with 
 \begin{equation}
\d{\cal F}\wedge\diamond{\cal F}= d \Big(\d{\cal A}\wedge\diamond{\cal F}\Big)- \d{\cal A}\wedge (-1)^{p-1} d\diamond{\cal F}. 
 \end{equation}
 Assume the variation of the potential ${\cal A}$ and the coframe $\vt^\a$ to be independent one on the other.
 Neglecting the total derivative term we are left with 
 the  field equation for ${\cal F}$ in the form
 \begin{equation}\label{3.1--8}
   \boxed{  d\diamond{\cal F}=0}
 \end{equation}
 
The second term in (\ref{3.1--5}) is identified as {\it Hilbert's energy-momentum current} of the field ${\cal F}$

\begin{equation}\label{3.1--9o}
\Sigma_\a=\frac{\delta {\cal L}}{\delta\vt^\a}
\end{equation}
Explicitly, we have 
\begin{equation}\label{3.1--9}
   \boxed{ \Sigma_\a=\frac 12\Big((-1)^p{\cal F}\wedge (e_\a\rfloor \diamond{\cal F})-(e_\a\rfloor {\cal F})\wedge\diamond{\cal F}\Big).}
\end{equation}
Observe that this expression is a straightforward consequence of the commutative identity (\ref{2.3--5}).
It is not depend on the field expression and even does not require the introduction of the potential. For an arbitrary untwisted $p$-form ${\cal F}$,  the expressions for $\Sigma_\a$ are twisted $(n-1)$-forms. Then they can serve as an invariant integrand over arbitrary  hypersurface in $\cal M$. 

Using the interior product identity
\begin{equation}
    e_\a\rfloor ({{\cal F}\wedge\diamond{\cal F}})=
     (e_\a\rfloor {{\cal F})\wedge\diamond{\cal F}}+
      (-1)^p {{\cal F}\wedge (e_\a\rfloor\diamond{\cal F}})
\end{equation}
we obtain two additional forms of the energy-momentum current expression 
\begin{equation}\label{3.1--10}
    \Sigma_\a=e_a\rfloor {\cal L} +
   (-1)^p{\cal F}\wedge (e_\a\rfloor \diamond{\cal F}).
\end{equation}
and
\begin{equation}\label{3.1--10x}
    \Sigma_\a=-e_\a\rfloor {\cal L}-(e_\a\rfloor {\cal F})\wedge \diamond{\cal F}.
\end{equation}
For an $n$-form $\cal F$, we have from (\ref{3.1--10}) a simple expression $\Sigma_\a=e_a\rfloor {\cal L}$. For a scalar form $\cal F$, Eq.(\ref{3.1--10x}) yields $\Sigma_\a=-e_a\rfloor {\cal L}$.

We summarize the results for the Maxwell-type model as follows:
\begin{itemize}
    \item The twisted $n$ form of an action  functional 
    \begin{equation}\label{sum-act}
        {\cal S}=-\frac 12 \int_M {\cal F}\wedge {\cal H}
    \end{equation}
    \item The untwisted $p$-form ${\cal F}$ and a twisted $(n-p)$-form ${\cal H}$ are connected by a linear  constitutive relation 
    \begin{equation}\label{sum-const}
        {\cal H}=\diamond{\cal F}
    \end{equation}
    \item The pair of field equations reads 
    \begin{equation}\label{sum-eqs}
        d{\cal F}=0,\qquad d{\cal H}={\cal J}
    \end{equation}
    \item The energy-momentum tensor is expressed as 
    \begin{equation}\label{sum-energy}
 \Sigma_\a=\frac 12\Big((-1)^p{\cal F}\wedge (e_\a\rfloor {\cal H})-(e_\a\rfloor {\cal F})\wedge{\cal H}\Big).
\end{equation}
\end{itemize}

\subsection{Energy-momentum current vs energy-momentum tensor}
Although the $(n-1)$-form current $\Sigma_\a$ is the proper description  of the energy-momentum quantity, it is useful to have an alternative tensor description. 
We will now derive a relationship between the energy-momentum current $\Sigma_\a$ and Hilbert's energy-momentum tensor $T^{ij}$. 

Recall that the current $\Sigma_\a$ is defined implicitly via the variation relation 
\begin{equation}\label{3.2--1}
    \d{\cal L}=\Sigma_\a\wedge \d\vt^\a.
\end{equation}
Then $\Sigma_\a$ is a vector-valued twisted $(n-1)$-form with $n^2$  independent components. 
The Lagrangian $n$-form ${\cal L}$ can be turned into a conventional Lagrangian function $L$ when a volume element on the manifold is prescribed.
 For this, we use a relation ${\cal L}=L\,{\rm vol}$. 
For a tensor representation of the energy-momentum current, we need a basis of twisted $(n-1)$-forms. It can be constructed from the set $\{e_\a\rfloor{\rm vol}\}$. Then the twisted form $\Sigma_\a$ is expressed as 
\begin{equation}\label{3.2--3}
    \Sigma_\a=\Theta_\a{}^\b \left(e_\b\rfloor{\rm vol}\right),
\end{equation}
or, equivalently, 
\begin{equation}\label{3.2--3x}
    \vt^\b\wedge\Sigma_\a=\Theta_\a{}^\b{\rm vol}.
\end{equation}
With respect to the smooth non-degenerate transformations of the coframe, the coefficients $\Theta_\a{}^\b$  constitute an  ordinary (untwisted) tensor. This  tensor   has the same $n^2$ independent components as the vector-valued $(n-1)$-form  $\Sigma_\a$. Note that the  definition of $\Theta_\a{}^\b$ requires only a  volume structure on the manifold.  Such structure is prescribed by an arbitrary smooth twisted $n$-form. In particular, $\Theta_\a{}^\b$ can be defined on a manifold without a prescribed metric. 

Let us turn now to Hilbert's energy-momentum tensor. It can also be defined implicitly using the variation relation
\begin{equation}\label{3.2-4}
    \d{\cal L}=\frac 12 T^{ij}\d g_{ij}{\rm vol}.
\end{equation}
When the volume element is identified as ${\rm vol}=\sqrt{-g}d^nx$, this definition corresponds to the widely-used variation derivative expression
\begin{equation}\label{3.2-4x}
    T^{ij}=\frac 2{\sqrt{-g}}\frac{\d L}{\d g_{ij}}
\end{equation}
The existence of a richer structure on the manifold—-metric structure is required by definition (\ref{3.2-4}). The tensor $T^{ij}$ is evidently symmetric  and composed of only $n(n-1)/2$ independent components. Furthermore, this tensor is referred to the holonomic coordinates $x^i$ basis (rather than the unholonomic basis $\vt^\a$) and transformed by the normal tensor law under the coordinate transformations.
Taking these restrictions into account,  we are able to establish a relation between two energy-momentum quantities.

\begin{prop}
Let a manifold be endowed with a metric $g_{ij}$ and an orthonormal coframe $\vt^\a=\vt^\a_idx^i$ such that 
\begin{equation}\label{3.2-5}
    g_{ij}=\eta_{\a\b}\vt^\a{}_i\,\vt^\b{}_j.
\end{equation}
The coframe energy-momenum current $\Sigma_\a$ and the Hilbert energy momentum tensor $T^{ij}$ are related as
\begin{equation}\label{3.2--6}
 T^{\a\b}{\rm vol}=\frac 12\Sigma_\g\wedge\left(\eta^{\a\g}\vt^\b+\eta^{\b\g}\vt^\a\right),
\end{equation}
where 
\begin{equation}\label{3.2-7}
T^{\a\b}=
T^{ij}\vt^\a{}_i\,\vt^\b{}_j.
\end{equation}
\end{prop}
\begin{proof}
It is convenient to express the variation of the Lagrangian in term of the variation matrix $\varepsilon^\a{}_\g$ satisfying the equation $\d\vt^\a=\varepsilon^\a{}_\g\vt^\g$. 
Substituting  expression (\ref{3.2--3}) into Eq.(\ref{3.2--1}) we have 
\begin{equation}\label{3.2-9}
    \d{\cal L}=\Theta_\a{}^\b \left(e_\b\rfloor{\rm vol}\right)\wedge \d\vt^\a.
\end{equation}
In terms of the variation matrix, it reads 
\begin{equation}\label{3.2-10}
      \d{\cal L}=-\Theta_\a{}^\b\varepsilon^\a{}_\g \vt^\g\wedge \left(e_\b\rfloor{\rm vol}\right).
\end{equation}
Using (\ref{2.1--9}) we obtain a compact expression
\begin{equation}\label{3.2-11}
      \d{\cal L}= -\Theta_\a{}^\b\varepsilon^\a{}_\b\,{\rm vol}.
\end{equation}

To deal with Hilbert's energy-momentum tensor, we must express the variation of the metric tensor in terms of the variation of the coframe. Using (\ref{3.2-4}) we write 
\begin{equation}\label{3.2-12}
    \d g_{ij}=\d(\eta_{\a\b}\vt^\a{}_i\,\vt^\b{}_j)=\eta_{\a\b}\left(\d(\vt^\a{}_i)\,\vt^\b{}_j+\vt^\a{}_i\,\d(\vt^\b{}_j)\right).
\end{equation}
In terms of the dual frame components $e_\rho{}^i$ satisfying $e_\rho{}^i\vt^\mu{}_i=\d^\mu_\rho$, we obtain 
\begin{equation}
    \d\vt^\mu=(\d\vt^\mu{}_i) dx^i=(\d\vt^\mu{}_i)e_\rho{}^i\vt^\rho.
\end{equation}
Thus
\begin{equation}
    \varepsilon^\mu{}_\rho\vt^\rho=(\d\vt^\mu{}_i)e_\rho{}^i\vt^\rho,
\end{equation}
or, equivalently,
\begin{equation}
   \d\vt^\mu{}_i =\varepsilon^\mu{}_\rho\vt^\rho{}_i.
\end{equation}
Consequently the variation of the metric tensor is presented in the form
\begin{equation}\label{3.2-12}
    \d g_{ij}=\eta_{\a\b}\left(\varepsilon^\a{}_\rho\,\vt^\rho{}_i\vt^\b{}_j+\varepsilon^\b{}_\rho\vt^\a{}_i\,\vt^\rho{}_j\right).
\end{equation}
Hence the metric variation of the Lagrangian (\ref{3.2-4}) is presented as
\begin{equation}\label{3.2-13}
    \d{\cal L}=\frac 12\eta_{\a\b}\left(\varepsilon^\a{}_\rho\,T^{\rho\b}+\varepsilon^\b{}_\rho T^{\a\rho}\right){\rm vol}=\eta_{\a\b}\varepsilon^\a{}_\rho\,T^{\rho\b}{\rm vol},
\end{equation}
where $T^{\a\b}$ is defined in (\ref{3.2-7}). 

We will now compare two expressions of  $ \d{\cal L}$ given in equations (\ref{3.2-11}) and (\ref{3.2-13}). Even though these expressions appear to be very similar, they have different features. Eq.(\ref{3.2-11}) is independent of the metric, whereas Eq.(\ref{3.2-13}) explicitly contains the metric tensor.   Moreover, the expression in (\ref{3.2-11}) contains  the asymmetric tensor $\Theta_\a{}^\a$, while the expression in (\ref{3.2-13}) contains the symmetric tensor $T^{\rho\b}$. As a result, the corresponding variations of the Lagrangian in  (\ref{3.2-13}) are only dependent on the symmetric variations of the coframe. Indeed, using the notation $\varepsilon_{\b\rho}=\eta_{\a\b}\varepsilon^\a{}_\rho$ we have from (\ref{3.2-13}) the symmetry relation $\varepsilon_{\b\rho}=\varepsilon_{\rho\b}$. Through using metric tensor, we rewrite now (\ref{3.2-11}) as $\d{\cal L}= -\Theta^{\a\b}\varepsilon_{\a\b}\,{\rm vol}$ and only employ the symmetric variation matrix in this expression.
As a result of comparing two expressions for $ \d{\cal L}$, we have
\begin{equation}
T^{\a\b}=-\frac 12(\Theta^{\a\b}+\Theta^{\a\b}).
\end{equation}
In terms of differential forms, this equation is expressed as
\begin{equation}
    T^{\a\b}\,{\rm vol}=-\frac 12
    (\eta^{\g\a}\vt^\b\wedge\Sigma_\g
+\eta^{\g\b}\vt^\a\wedge\Sigma_\g).
\end{equation}
The latter expression is equivalent to (\ref{3.2--6}). 
\end{proof}

\subsection{Trace of Hilbert's energy-momentum current}
The trace of the energy-momentum  current $\Theta_\a{}^\a$ prescribes in  field theory  essential physical quantity, such as photon mass. In differential-form approach, we can construct a scalar-valued $n$-form $\vt^\a\wedge\Sigma_\a$ that has one component as the scalar $\Theta_\a{}^\a$. 
Using Eq.(\ref{3.1--2}), we have
\begin{equation}\label{3.3--1}
    \vt^\a\wedge\Sigma_\a=\Theta_\a{}^\b \vt^\a\wedge\left(e_\b\rfloor{\rm vol}\right)=\Theta_\a{}^\a\,{\rm vol}.
\end{equation}
 Consequently, the trace of the energy-momentum tensor  is equivalent to the $n$-form $\vt^\a\wedge\Sigma_\a$. 

\begin{prop} The trace term for a Maxwell-type Lagrangian with an arbitrary $p$-form strength  ${\cal F}$ is expressed as 
\begin{equation}\label{3.3--2}
    \vt^\a\wedge\Sigma_\a=(2p-n){\cal L}.
\end{equation}
In particular, the energy momentum current is traceless if and only if the space is even-dimensional and ${\cal F}$ is a middle form. 
\end{prop}
\begin{proof}
Calculate using (\ref{3.1--2}) 
\begin{eqnarray}\label{3.3--3}
    \vt^\a\wedge\Sigma_\a&=&\frac 12\vt^\a\wedge\Big((-1)^p{\cal F}\wedge (e_a\rfloor \diamond{\cal F})
    -(e_a\rfloor {\cal F})\wedge\diamond{\cal F} 
   \Big)\nonumber\\
   &=&\frac 12\Big[({\cal F}\wedge \vt^\a\wedge(e_a\rfloor \diamond{\cal F})-
   \vt^\a\wedge(e_a\rfloor {\cal F})\wedge\diamond{\cal F}
   \Big]
\end{eqnarray}
Using the first of the relations (\ref{2.1--8}) we obtain 
\begin{equation}\label{3.3--4}
 \vt^\a\wedge(e_a\rfloor {\cal F})=p{\cal F}   
\end{equation}
and 
\begin{equation}\label{3.3--5}
 \vt^\a\wedge(e_a\rfloor \diamond{\cal F})=(n-p){\cal F}. 
\end{equation}
Consequently, 
\begin{equation}\label{3.3--6}
 \vt^\a\wedge\Sigma_\a=\frac 12(n-2p){\cal F}\wedge\diamond {\cal F}=-(n-2p){\cal L}.
\end{equation}
\end{proof}
In important physics applications, such as electrodynamics, the field strength ${\cal F}$ is a 2-form in the 4-dimensional space-time. Eq.(\ref{3.3--2}) shows that the energy-momentum current for such models is traceless. 
\subsection{Symmetry of Hilbert's energy-momentum current}
The energy-momentum current $\Sigma_\a$ is equivalent to a mixed tensor  $T_\a{}^\b$ with $n^2$ components. The metric tensor $\eta^{\a\b}$ must be used to convert this tensor into a covariant tensor $T^{\a\b}$ or a contravariant tensor $T_{\a\b}$ in order to examine its symmetry. Then it is possible to extract the symmetry part $T^{(\a\b)}$  of the $n(n+1)/2$ components and the skew-symmetric part $T^{[\a\b]}$ of the $n(n-1)/2$ components.
The same $n(n-1)/2$ independent components can be represented equivalently as a scalar-valued $(n-2)$-form. We define 
\begin{equation}
    {\cal W}=\eta^{\a\b}e_\a\rfloor\Sigma_\b.
\end{equation}
\begin{prop}
The Hilbert energy-momentum current $\Sigma_\a$ represents the symmetric tensor iff 
the $(n-2)$-form ${\cal W}$ vanishes.
\end{prop}
\begin{proof}
Using the tensor presentation (\ref{3.2--3}), we derive
\begin{equation}
    {\cal W}=\eta^{\a\b}e_\a\rfloor \left(\Theta_\b{}^\nu e_\nu\rfloor{\rm vol}\right)=\Theta^{\a\nu}\left(e_\a\rfloor  e_\nu\rfloor{\rm vol}\right)=\Theta^{[\a\b]}\left(e_\a\rfloor  e_\b\rfloor{\rm vol}\right)
\end{equation}
Thus the $(n-2)$-form ${\cal W}$ indeed represents the skew-symmetric part of the energy-momentum tensor. The equation ${\cal W}=0$ means symmetry of the energy-momentun current. 
\end{proof}

\begin{prop}
For the Maxwell-type Lagrangian ${\cal L}=(1/2) {\cal F}\wedge\diamond{\cal F}$, the skew-symmetric 2-form ${\cal W}$ of the Hilbert energy-momentum current $\Sigma_\a$ is expressed as 
\begin{equation}\label{W}
    {\cal W}= (-1)^p \eta^{\a\b}\Big((e_\a\rfloor  {\cal F})\wedge (e_\b\rfloor \diamond{\cal F})\Big)
\end{equation}
\end{prop}
\begin{proof}
Using (\ref{3.1--9}) we have
  \begin{equation}
    {\cal W}=\frac 12\eta^{\a\b} e_\a\rfloor  \Big((-1)^p{\cal F}\wedge (e_\b\rfloor \diamond{\cal F})-(e_\b\rfloor {\cal F})\wedge\diamond{\cal F}\Big)
  \end{equation} 
Due to identity (\ref{2.1--4}) we are left with 
 \begin{equation}
    {\cal W}=\frac 12\eta^{\a\b}  \Big((-1)^p{(e_\a\rfloor \cal F})\wedge (e_\b\rfloor \diamond{\cal F})-(-1)^{p-1}(e_\b\rfloor {\cal F})\wedge (e_\a\rfloor \diamond{\cal F})\Big)
  \end{equation} 
that is equivalent to (\ref{W}).

\end{proof}

\section{Applications}
In this part, we will look at certain applications of the formalism described above. We will be dealing with viable physics models in the ordinary fourth-dimensional space-time. Our aim is to discuss different types of the dual map. In all models, the energy-momentum current follows from the same expression (\ref{3.1--9}).  In the first pair of examples, we are dealing with the well-known  models in vacuum. In this case, the dual map is proportional to the ordinary Hodge map.  Particularly, in the tangent space, we employ the conventional version of the Minkowski metric $\eta_{\a\b}={\rm diag}(+1,-1,-1,-1)$. 
In the second pair of examples, we consider premetric models of electromagnetism and gravity. 
\subsection{Scalar field}
For a complex scalar field $\varphi $ on a 4-dimensional manifold, the Lagrangian 4-form can be expressed using the differential forms formalism. We assume  
\begin{equation}\label{4.1}
    {\cal L}=\frac 12 d\varphi\wedge *\,\overline{d\varphi}+\frac 12m^2\varphi*\,\overline{\varphi},
\end{equation}
where the star $*$ denotes Hodge's operator while the bar states for the complex conjugate. Both of these two operators are defined for differential forms of arbitrary order. 
The only forms for which we must demand their definition in (\ref{4.1}) are those of the order $p=0,1$. 
Two terms in (\ref{4.1}) can be regarded as generalized Maxwell-type Lagrangians (\ref{sum-act}). The generalized Hodge map can be constructed in relation to the general scheme as a product of Hodge's dual and complex conjugate operators, $\diamond w=*\,\overline w$.

To illustrate how the Lagrangian 4-form (\ref{4.1})  relates to the standard tensor representation, we can express the field strength in a non-holonomic and a holonomic (coordinate) basis, respectively, as follows:
\begin{equation}
    d\vp=\vp_\mu\vt^\mu=\vp_{,i}dx^i\qquad {\rm with}\qquad \vp_{,i}=\frac{\partial \vp}{\partial x^i}=\vp_\mu\vt^\mu_i.
\end{equation}
The coframe $\vt^\mu$ is assumed to be  pseudo-orthonormal with respect to the metric tensor $ g_{ij}=\eta_{\mu\nu}\vt^\mu_i\vt^\nu_j$.

The Lagrangian 4-form is then expressed in standard form ${\cal L}=L{\rm vol}$, with
\begin{equation}\label{4.1x}
    L=\frac 12 \eta^{\mu\nu}\vp_\mu\overline{\varphi}_\nu+\frac 12 m^2\varphi\,\overline{\varphi}=\frac 12g^{ij} \vp_{,i}\overline{\varphi}_{,j} +\frac 12 m^2\varphi\,\overline{\varphi}.
\end{equation}
This Lagrangian scalar can only be used to derive the energy-momentum tensor when the factor $\sqrt{-g}$ is included in the equation.
However, we will move forward with the differential form variation that was previously described.
The field equations (\ref{3.1--8}) takes the form 
\begin{equation}\label{4.1--1}
    d*d\overline{\varphi}=m^2*\overline{\varphi}.
\end{equation}
This equation is applicable also on a curved manifold. 

Let us turn to the energy-momentum tensor.  For the first term in (\ref{4.1}) we can identify the field strength as  ${\cal F}=d\vp$ of the order $p=1$. We choose the constitutive map acted on 1-form ${\cal F}$ as $\diamond {\cal F}=-*\overline{{\cal F}}$. Then the Lagrangian takes the canonical form ${\cal L}=-(1/2){\cal F}\wedge \diamond {\cal F}$ and the general expressions (\ref{3.1--9},\ref{3.1--10},\ref{3.1--10x}) for the energy-momentum current are applicable. We have the first part of the energy-momentum current (\ref{sum-energy}
\begin{equation}\label{3.1--9scalar}
    {}^{(1)} \Sigma_\a=\frac 12\Big({ d\vp}\wedge (e_\a\rfloor *d\overline \vp)+(e_\a\rfloor {d\vp})\wedge*d\overline{\vp}\Big).
\end{equation}
In components, this expression takes the form
\begin{equation}\label{4.1--2y}
    {}^{(1)} \Sigma_\a=\frac 12 \vp_\mu\overline \vp_\nu\left(2\d^\mu_\a\eta^{\b\nu}-\d^\b_\a\eta^{\mu\nu}\right)e_\b\rfloor{\rm vol},
\end{equation}
 where the the standard energy-momentum expression for the scalar field is visible.  
 
  Even the second term in (\ref{4.1}) can be formally brought into the canonical form  ${\cal L}=-(1/2){\cal F}\wedge \diamond {\cal F}$. For this we can choose the zeroth-order``strength"  ${\cal F}=\vp$ and the dual map $\diamond {\cal F}=*\overline{{\cal F}}$. Thus we can use the energy-momentum current expression (\ref{sum-energy}) once more. 
We have 
 \begin{equation}\label{4.1--2x}
{}^{(2)} \Sigma_\a=
   \frac 12 m^2\varphi\overline{\varphi}\,(e_\a\rfloor {\rm{vol}}).
\end{equation}
The total energy-momentum current is expressed as a sum 
$\Sigma_\a={}^{(1)} \Sigma_\a+{}^{(2)} \Sigma_\a$.

The trace of the the energy-momentum current is determined by the product $\vt^\a\wedge \Sigma_\a$. We obtain  
\begin{equation}
    \vt^\a\wedge \Sigma_\a=
 ( -\eta^{\mu\nu}\vp_\mu\overline \vp_\nu+2 m^2\varphi\overline{\varphi})\,{\rm{vol}}.
\end{equation}

Note, that the definition of  Lagrangian (\ref{4.1}), the field equation (\ref{4.1--1}), and the energy-momentum currents  (\ref{4.1--2y}), (\ref{4.1--2x}) do  not require the  $\diamond$-operator to be defined on forms of order $p\ne 0,1$.

\subsection{Electromagnetism in vacuum}
The electromagnetic field in vacuum is described on a 4-dimensional manifold by a Lagrangian of the type (\ref{3.1--2}) with the general linear operator proportional to the Hodge dual operator $\diamond=*$.  

The field strength ${\cal F}$ is an untwisted 2-form. It is assumed to be exact, i.e., ${\cal F}$ is expanded in term of  the potential 1-form ${\cal A}$ such that 
\begin{equation}
    {\cal F}=d{\cal A}.
\end{equation}   Using the phenomenological source---twisted 3-form of the electric current ${\cal J}$, we express a Lagrangian 4-form as
\begin{equation}
    {\cal L}=-\frac 12{\cal F}\wedge*{\cal F}-{\cal A}\wedge {\cal J}
\end{equation}
In this setting, the potential 1-form ${\cal A}$ is assumed to be a dynamical variable, the coframe $\vt^\a$ appearing implicitly in the definition of Hodge's map is a semi-dynamical variable, while the 3-form source ${\cal J}$ is  non-dynamical variable. 
Applying variations of the Lagrangian with respect to the potential  ${\cal A}$ and to the coframe field $\vt^\a$, we obtain the field equation 
\begin{equation}
    d*{\cal F}={\cal J}
\end{equation}
and the energy-momentum current
\begin{equation}
    \Sigma_\a=\frac 12\Big((e_a\rfloor {\cal F})\wedge*{\cal F} -
   (e_a\rfloor {\cal *F})\wedge{\cal F}\Big)
\end{equation}
Since the strength ${\cal F}$ acting on 4-dimensional manifold, this current is traceless, due to Proposition 4. It is also symmetric---the equation ${\cal W}=0$ can be proved using the algebraic properties of the Hodge map, see \cite{Itin:2001xz}.

Substituting into Eq.(\ref{3.1--10x}) the basis expressions 
\begin{equation}\label{3.1--3}
{\cal F}=
\frac 12 F_{\mu\nu}\vt^\mu\wedge\vt^\nu,\qquad  
\qquad  \diamond{\cal F}=\frac 12 H^{\rho\sigma}
(e_\rho\rfloor e_\sigma\rfloor {\rm vol}
\end{equation}
we obtain the familiar expression of the energy-momentum tensor 
\begin{equation}
    T_\a{}^\b =L\d_\a^\b-\frac 12  F_{\a\mu}H^{\mu\b}
\end{equation}
Note that this expression is  valid not only in the standard Maxwell's  electrodynamics on a flat space but also on a curved manifold.

\subsection{Premetric electromagnetism}
Premetric electromagnetism model is described by two second-order differential forms: an untwisted form ${\cal F}$ of the {\it  field strength} and  a twisted form ${\cal H}$ of the {\it electromagnetic excitation}. Both forms are defined on a differential 4-dimensional manifold without a prescribed metric. The form ${\cal F}$ is assumed to be exact, i.e., an untwisted 1-form ${\cal A}$ of {\it potential} is defined 
\begin{equation}
    {\cal F}=d{\cal A}. 
\end{equation}
The Lagrangian is assumed in the standard Maxwell-type form
\begin{equation}
    {\cal L}=\frac 12{\cal F}\wedge{\cal H}-{\cal A}\wedge {\cal J}. 
\end{equation}
Here ${\cal J}$ is a twisted 3-form of {\it electric current}. It is  a non-dynamical phenomenological quantity that is not included in the variation procedure. 

Two 2-order differential forms ${\cal F}$ and ${\cal H}$ are not independent. In a wide range of field strength, they are related by linear {\it constitutive relation}
\begin{equation}\label{cons-rel}
    {\cal H}=\kappa {\cal F},
\end{equation}
where $\kappa$ is a linear operator. It is assumed to satisfy the defining properties of the generalized Hodge map listed in Definition 3. Note that $\kappa$ is required to  be defined on the 2-forms only. The commutative formula (\ref{2.3--5}) holds even in this restricted case. 

To express the operator $\kappa$ explicitly, we expand the 2-forms  in the coframe basis 
\begin{equation}
    {\cal F}=\frac 12 F_{\a\b}\vt^\a\wedge \vt^\b,\qquad 
    {\cal H}=\frac 12 H_{\a\b}\vt^\a\wedge \vt^\b.
\end{equation}
In this description, the $F_{\a\b}$ are components of an ordinary tensor, while $H_{\a\b}$ are components of a pseudo-tensor. 
The operator $\kappa$ is defined as a pseudo-tensor related the components of the field strength to the components of excitation form 
\begin{equation}
    H_{\a\b}=\frac 12 \kappa_{\a\b}{}^{\g\d}F_{\g\d}\,.
\end{equation}
Both forms are assumed to be related to the  same basis. 

Since $\kappa_{\a\b}{}^{\g\d}$ is skew-symmetric in two pairs of its indices it is convenient to deal with an equivalent pseudo-tensor 
\begin{equation}
    \chi^{\a\b\g\d}=\kappa_{\mu\nu}{}^{\g\d}\varepsilon^{\a\b\mu\nu}
\end{equation}
Here $\varepsilon^{\a\b\mu\nu}$ is a permutation pseudo-tensor.
Observe the symmetry relations
\begin{equation}
    \chi^{\a\b\g\d}=-\chi^{\b\a\g\d}=-\chi^{\a\b\d\g}\,.
\end{equation}
Consequently the pseudo-tensor $\chi^{\a\b\g\d}$ has 36 independent components. It can be irreducibly decomposed \cite{Hehl-Obukhov} into a sum of three independent sub-tensors: the principal  part of 20 components, the skewon part of 15 components, and the axion part of 1 component. The skewon part does not contribute to the Lagrangian, while the axion part does not contribute in Hilbert's energy-momentum tensor \cite{Hehl-Obukhov}. Thus we restrict ourselves  to the principle part of the constitutive tensor of 20 independent components. Note that all these components are observable in solid-state physics. As a result the constitutive tensor is assumed to satisfy the additional symmetry relations
\begin{equation}
    \chi^{\a\b\g\d}=\chi^{\g\d\a\b},\qquad {\rm and }\qquad \chi^{[\a\b\g\d]}=0\,.
\end{equation}

The constitutive relation (\ref{cons-rel}) can be viewed as a linear map between two spaces of the second-order differential forms.  

Although, the constitutive relation (\ref{cons-rel})  is  defined for the second-order differential forms only, the commutation formula can be applied also in this case. Consequently, we are coming to the field equation (\ref{sum-eqs})  
\begin{equation}
    d{\cal H}=J
\end{equation}
and the energy-momentum current (\ref{sum-energy}) 
\begin{equation}
     \Sigma_\a=\frac 12\Big((e_a\rfloor {\cal F})\wedge\kappa{\cal F} -
   {\cal F}\wedge (e_a\rfloor \kappa{\cal F})\Big).
\end{equation}
In term of the forms ${\cal F}$ and ${\cal H}$, it reads
\begin{equation}
     \Sigma_\a=\frac 12\Big((e_a\rfloor {\cal F})\wedge{\cal H} -
   {\cal F}\wedge (e_a\rfloor {\cal H})\Big).
\end{equation}
This current is traceless due to the general fact given in Proposition 4. 

In \cite{Hehl-Obukhov}, the energy-momentum current expression was postulated as an independent axiom based on an extension of the vacuum formula. Here we derived it from the Lagrangian by the variation relation.

\subsection{Premetric gravity}
The premetric gravity model \cite{Itin:2016nxk},  \cite{Itin:2018dru} is based on two Maxwell-type equations 
\begin{equation}\label{pr-gr-1}
    d{\cal H}_\a=\Sigma_\a\qquad d{\cal F}^\a=0.
\end{equation}
Here ${\cal H}_\a$ is a covector-valued twisted 2-form of {\it gravitational excitation} , while ${\cal F}^\a$ is a vector-valued untwisted 2-form of {\it gravitational field strength}. $\Sigma_\a$ is a 3-form of energy-momentum current. 
In topologically good area, the homogeneous field equation can be resolved in term of potential 
\begin{equation}\label{pr-gr-2}
    d{\cal F}^\a=0 \quad\Longrightarrow\quad {\cal F}^\a=d\vt^\a.
\end{equation}
Since the set of 1-forms $\vt^\a$ is defined up to a total derivative, they can be chosen linearly independent. Then $\vt^\a$ not only a potential but also a reference basis. 
We express the 2-forms an this basis
\begin{equation}\label{pr-gr-3}
    {\cal F}^\a=\frac 12 F^\a{}_{\mu\nu}\vt^\mu\wedge\vt^\nu,\qquad {\cal H}_\a= \frac 12H_\a{}^{\mu\nu}e_\mu\rfloor e_\nu\rfloor{\rm vol}.
\end{equation}
Here $e_\a$ is a frame dual to $\vt^\a$, i.e., $e_\a\rfloor\vt^\b=\d_\a^\b$. 

The tensors $F^\a{}_{\mu\nu}$ and $H_\a{}^{\mu\nu}$ are assumed to be connected by linear homogeneous constitutive relation
\begin{equation}\label{pr-gr-4}
    H_\a{}^{\mu\nu}=\frac 12 \chi_\a{}^{\mu\nu}{}_\b{}^{\rho\sigma}{}F^\b{}_{\rho\sigma} 
\end{equation}
The constitutive tensor $\chi_\a{}^{\mu\nu}{}_\b{}^{\rho\sigma}$ is skew-symmetric in the pairs of its upper indices. 
\begin{equation}\label{pr-gr-5}
    \chi_\a{}^{\mu\nu}{}_\b{}^{\rho\sigma}=-\chi_\a{}^{\nu\mu}{}_\b{}^{\rho\sigma}=-\chi_\a{}^{\mu\nu}{}_\b{}^{\sigma\rho}. 
\end{equation}
Consider a Lagrangian  of the Maxwell-type form
\begin{equation}\label{pr-gr-6}
        {\cal L}=-\frac 12 {\cal F}^\a\wedge{\cal H}_\a +{}^{\rm (mat)}{\cal L}(\psi,d\psi,\vt^\a),
\end{equation}
where $\psi$ represents the matter field. 

Since the Lagrangian satisfied the condition of the general Maxwell-type action, 
the field equation takes the form (\ref{sum-eqs}) 
\begin{equation}\label{prem-eq}
    d{\cal H_\a}={}^{\rm (gr)}\Sigma_\a +{}^{\rm (mat)}\Sigma_\a. 
\end{equation}
Here the matter energy-momentum current is defined as 
\begin{equation}
    {}^{\rm (mat)}\Sigma_\a=
    \frac\d{\d\vt^\a}{ {}^{\rm (mat)}{\cal L}(\psi,d\psi,\vt^\a)}, 
\end{equation}
while the gravity-energy-momentum current is given in the form(\ref{sum-energy})
\begin{equation}
    {}^{\rm (gr)}\Sigma_\a=\frac 12\Big((e_a\rfloor {\cal F^\b})\wedge{\cal H}_\b -
   {\cal F^\b}\wedge (e_a\rfloor {\cal H}_\b)\Big).
\end{equation}
When a constitutive pseudotensor $\chi_\a{}^{\mu\nu}{}_\b{}^{\rho\sigma}$ is restricted to a special metric form \cite{Itin:2018dru} and Eq.(\ref{prem-eq}) is rewritten as 
\begin{equation}
    d{\cal H_\a}-{}^{\rm (gr)}\Sigma_\a ={}^{\rm (mat)}\Sigma_\a. 
\end{equation}
it turns to be equivalent to the standard Einstein's equation.  
\section{Conclusion}
In this paper, we are looking for an extension of Hilbert's definition of the energy-momentum tensor to  premetric field models. In this case, the metric tensor is replaced by a general constitutive law that relates the basic fields in the model. We explain how  the variation procedure can be applied on fields presented by differential form instead of separated components. In this formalism, the commutative relation for the mapped forms is derived. Then we express the Lagrangian density in term of Hodge-type dual map and derive the corresponding energy-momentum current (twisted vector-valued 3-form). This expression turns out to be straightforwardly related to the commutative relation. The applications of the energy-momentum current formula to specific field models yields the expressions appearing in literature. In the case of a metric constitutive map, these formulas was derived from a Lagrangian. For pure premetric models,  such formulas was previously postulated. 

The applications of the current results can be used in more complicated field models such as non-local electrodynamics \cite{Mashhoon}, non-local gravity \cite{Hehl-Mashhoon}, non-linear electrodynamics  \cite{Obukhov:2002xa}, and Finsler modified electrodynamics \cite{Itin:2014uia}.

\end{document}